\documentclass{article}
\usepackage[utf8]{inputenc}
\usepackage{preamble/variables}
\arxivtrue
\usepackage{preamble/pat}
\usepackage{preamble/macros}

\newtheorem{observation}{Observation}

\title{Improved Spanning on Theta-5}
\author{%
	Prosenjit Bose\footnotemark[1]\,\,\footnotemark[2]
\and
	Darryl Hill\footnotemark[3]\,\,\footnotemark[2]
\and
	Aurélien Ooms\footnotemark[4]\,\,\footnotemark[5]
}

\begin{document}
\renewcommand{\thefootnote}{\fnsymbol{footnote}}
\maketitle
\footnotetext[1]{Carleton University, Ottawa, Canada. Email: {\tt jit@scs.carleton.ca}}
\footnotetext[2]{Research supported in part by NSERC}
\footnotetext[3]{Carleton University, Ottawa, Canada. Email: {\tt darrylhill@cunet.carleton.ca}}
\footnotetext[4]{Email: {\tt aurelien.ooms@gmail.com}}
\footnotetext[5]{%
Supported by the VILLUM Foundation grant 16582.
Part of this research was accomplished while the author was a PhD
student at ULB under FRIA Grant 5203818F (FNRS).%
}

\begin{abstract}
	We show an upper bound of \( \frac{
			\sin\left(\frac{3\pi}{10}\right)
		}{
			\sin\left(\frac{2\pi}{5}\right)-\sin\left(\frac{3\pi}{10}\right)
		}
		<5.70\) on the spanning ratio of
	$\Theta_5$-graphs,
	improving on the previous best known upper bound of \(9.96\)
	[Bose, Morin, van Renssen, and Verdonschot.
	The Theta-5-graph is a spanner.
	\emph{Computational Geometry}, 2015.%
	]
	
\end{abstract}

\section{Introduction}

A geometric graph $G$ is a graph whose vertex set is a set of points $P$ in the
plane, and where the weight of an edge $uv$ is equal to the Euclidean distance
$|uv|$ between $u$ and $v$.
Informally, a $\Theta_k$-graph
is a geometric graph built by dividing the area around each point of
\(v \in P\)
into \(k\) equal angled cones, connecting \(v\) to the \emph{closest} neighbor in each
cone (we shall define closest later).
Such graphs arise naturally in settings like wireless
networks, where signals to anyone but your nearest neighbor may be
drowned out by interference. Moreover, the fact that
signal strength fades quadratically with distance, and
thus that power requirements are proportional to the square of the distance the
signal has to travel, makes many small hops economically superior to one
large hop, even if the sum of the distances is larger.
The \emph{spanning ratio} (sometimes called the \emph{stretch factor})
of a geometric graph \(G\) is the maximum
over all pairs \(u,v \in P\) of the ratio between the length of the shortest
path from \(u\) to \(v\) in \(G\) and the Euclidean distance from \(u\) to
\(v\).
Using simple geometric observations and techniques,
we give a new analysis of the spanning ratio of \(\Theta_5\)-graphs, bringing
down the best known upper bound from
$9.96$~\cite{sander5} to $5.70$.
\state{TheoremMain}{theorem}{\label{main}
	Given a set \(P\) of points in the plane,
	the \(\Theta_5\)-graph of \(P\) is a \(5.70\)-spanner.
}

\begin{figure}
	\centering
	\begin{subfigure}[b]{0.45\textwidth}
		\centering
		\includegraphics[page=8,width=\textwidth]{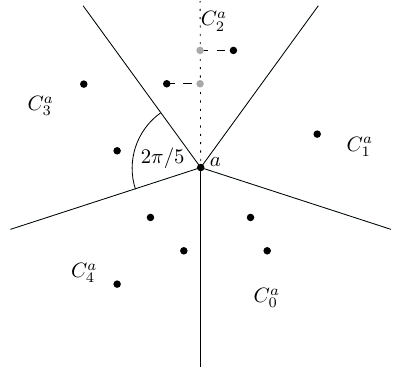}
		\caption{Measure of the distance to point $a$.}%
		\label{fig:construction-theta-5}
	\end{subfigure}%
	\hfill
	\begin{subfigure}[b]{0.45\textwidth}
		\centering
		\includegraphics[page=9,width=\textwidth]{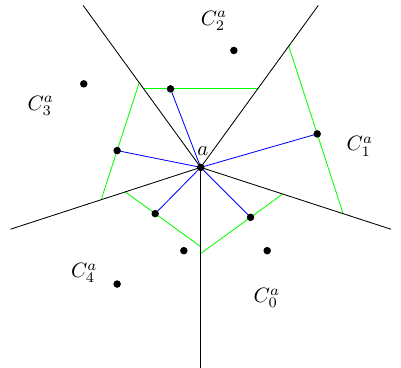}
		\caption{The neighbors of \(a\) in the \(\Theta_5\)-graph.}
	\end{subfigure}
	\caption{The area around a point $a$ is divided into cones with angle $2\pi/5$.}\label{construction}
\end{figure}

$\Theta_k$-graphs were introduced simultaneously by Keil and
Gutwin~\cite{keiltheta,keil1992}, and Clarkson~\cite{clarkson1987}. Both papers
gave a spanning ratio of $1/(\cos\theta - \sin\theta)$, where $\theta=2\pi/k$
is the angle defined by the cones.
Observe that this gives a constant spanning ratio for $k\geq9$.
When this ratio \(t\) is constant, we call the graph a \(t\)-spanner.
Ruppert and Seidel~\cite{ruppert} improved this to
$1/(1-2\sin(\theta/2))$, which applies to $\Theta_k$-graphs with $k\geq 7$.
Chew~\cite{chewtd} gave a tight bound of $2$ for $k=6$.
Bose, De Carufel, Morin, van Renssen, and
Verdonschot~\cite{morenotbetter} give the current best bounds on the spanning
ratio of a large range of values of $k$.
For $k=5$, Bose, Morin, van Renssen, and Verdonschot~\cite{sander5} showed an
upper bound of $9.96$, and a lower bound of $3.78$.
For $k=4$, Bose, De Carufel, Hill, and Smid~\cite{soda19} showed a spanning
ratio of $17$, while Barba, Bose, De Carufel, van Renssen, and
Verdonschot~\cite{barbat4} gave a lower bound of $7$ on the spanning ratio.
For $k=3$,
although
Aichholzer, Bae, Barba, Bose, Korman, van Renssen, Taslakian, and
Verdonschot~\cite{theta3connected}
showed \(\Theta_3\) to be connected,
El Molla~\cite{nawarphd} showed that there is no constant $t$ for which
$\Theta_3$ is a $t$-spanner.

In this paper we study the spanning ratio of $\Theta_5$.
We consider two arbitrary vertices, $a$ and $b$, and show that there must exist
a short path between them using induction on the rank of the Euclidean distance
$|ab|$ among all distances between pairs of points in $P$.
Our main result states that
for all \(a, b \in P\)
the shortest path $\pathF{a,b}$ has
length $|\pathF{a,b}| \leq K \cdot |ab|$, where $K=5.70$.

Much of the difficulty in bounding the spanning ratio of the $\Theta_5$-graph stems from the following.
\begin{enumerate}
	\item The regular pentagon is not centrally symmetric.
	\item Give two vertices $a$ and $b$, it may be the case that every vertex $v$ adjacent to $a$ has the property that $|vb|>|ab|$. In other words, all the neighbours of $a$ are farther from $b$ than $a$ itself. 
\end{enumerate}

We organize the rest of the paper as follows.
In Section~\ref{prelim} we introduce concepts and notation, and give some
assumptions about the positions of $a$ and $b$ that do not reduce the
generality of our arguments.
In Section~\ref{casebycase} we solve all but a handful of cases using general
arguments that simplify the analysis. The remaining cases
are solved using ad-hoc methods, showing a spanning ratio of $K=6.16$.
In Section~\ref{down2K} we observe that only a single case requires $K \geq
6.16$.
We analyze this case in detail to show that $|\pathF{a,b}|\leq K\cdot |ab|$ for
all $K\geq5.70$.
\ifconclusion%
In Section~\ref{conclusion} we discuss directions for future work.%
\fi

\section{Preliminaries}\label{prelim}

Let $k\geq 3$ be an integer.
 Let \(P\) be set of points in the plane in general position, that is, all
distances (as defined below) between pairs of points are unique and no two points have the same $x$-coordinate or $y$-coordinate.
Construct the $\Theta_k$-graph of \(P\) as follows.
\aurelien{Is this general position assumption necessary? Can we not break ties
	in a consistent way? Should we say something like ``While this assumption makes
	the analysis simpler, it is not difficult to see that ties can be broken
	consistently when \(P\) is degenerate.''?}
%
The vertex set is \(P\).
For each $i$ with $0 \leq i < k$, let $\mathcal{R}_i$ be the ray
emanating from the origin that makes an angle of $2 \pi i /k$ with the negative
$y$-axis.\footnote{Angle values are given counter-clockwise unless otherwise stated.}
All indices are manipulated mod $k$, i.e., $\mathcal{R}_k = \mathcal{R}_0$.
For each vertex $v$ we add at most $k$ outgoing edges as follows:
For each $i$ with $0 \leq i < k$, let $\mathcal{R}_i^v$ be the ray emanating
from $v$ parallel to $\mathcal{R}_i$. Let $C_i^v$ be the cone consisting of all
points in the plane that are strictly between the rays $\mathcal{R}_i^v$ and
$\mathcal{R}_{i+1}^v$ or on $\mathcal{R}_{i+1}^v$. If $C_i^v$ contains at least
one point of $P \setminus \{v\}$, then let $w_i$ be the \emph{closest} such point to
$v$, where we define the closest point to be the point whose perpendicular
projection onto the bisector of $C_i^v$ minimizes the Euclidean distance to \(v\).
We add the directed edge $vw_i$ to the graph.
While the use of directed edges better illustrates this construction,
in what follows we regard all edges of a $\Theta_5$-graph as undirected. See Fig. \ref{construction} for an example of cones and construction.

\begin{figure}\centering
	\begin{subfigure}[b]{0.45\textwidth}
		\centering
		\includegraphics[page=1,width=\textwidth]{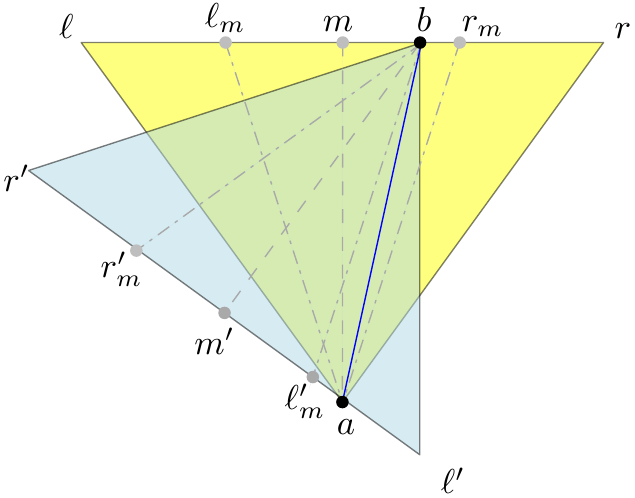}
		\caption{Assume $b$ is in $C_2^a$ and $a$ is in $C_4^b$.}%
		\label{examplet5}
	\end{subfigure}%
	\hfill
	\begin{subfigure}[b]{0.45\textwidth}
		\centering
		\includegraphics[page=2,width=\textwidth]{figures/generalcase.pdf}
		\caption{The angle $\alpha$.}%
		\label{alpha}
	\end{subfigure}
	\caption{Vertices $a$ and $b$ and the canonical triangles $T_{ab}$ and $T_{ba}$.}\label{layout}
\end{figure}

For the following description, refer to Fig. \ref{layout}. Consider two vertices $a$ and $b$ of $P$.
Given the \(\Theta_5\)-graph of \(P\),
we define the \emph{canonical triangle} $T_{ab}$ to be the triangle bounded by
the sides of the cone of $a$ that contains $b$ and the line through $b$ perpendicular to the
bisector of that cone.
Note that for every pair of vertices $a$ and $b$ there are two corresponding canonical triangles, namely $T_{ab}$ and $T_{ba}$.
Without loss of generality assume
that $b$ is in $C_2^a$. Let $\ell$ be the leftmost vertex of the triangle $T_{ab}$ and let
$r$ be the rightmost vertex of the triangle $T_{ab}$. Let $m$ be the midpoint of $\ell r$.
Note that $a$ must be in $C_4^b$ or $C_0^b$; since the cases are symmetric in what follows, without loss of generality we
consider the case where $a$ is in $C_4^b$. Thus $b$ is to the right of $m$. Let
$r_m$ be the intersection of $\ell r$ and the bisector of $\angle ram$%
\footnote{%
In what follows we use \(\triangle abc\) to denote the triangle defined by the
points \(a\), \(b\), and \(c\) (given counter-clockwise).
We use \(\angle abc\) to denote the amplitude of the angle at \(b\) in that
triangle.%
}%
, and let
$\ell_m$ be the intersection of $\ell r$ and the bisector of $\angle ma\ell$.
Let $\ell'$ and $r'$ be the left and right endpoints of $T_{ba}$ respectively (as seen from $b$ facing $a$).
Let $m'$ be the midpoint of $\ell'r'$, and let $\ell_m'$ and $r_m'$ be the
intersections of $\ell'r'$ and the bisector of $\angle \ell'bm'$ and $\angle
m'br'$ respectively. See Figure~\ref{examplet5}. Let $\alpha = \angle bam$
and let $\alpha' = \angle abm'$. Note that $\alpha +\alpha' = \pi/5$ since
$\alpha$ and $ \frac{2 \pi}{5} - \alpha'$ are alternate interior angles.
Thus either $\alpha\leq
\pi/10$ or $\alpha'\leq \pi/10$. Without loss of generality, we assume $\alpha
\leq \pi/10$. Let $c$ be the closest neighbor to $a$ in $C_2^a$, and let $d$
be the closest neighbor to $b$ in $C_4^b$. See Figure~\ref{alpha}.
For simplicity, we write ``$\Theta_5$''
to mean ``the $\Theta_5$-graph of $P$''.

To sum up our assumptions following this discussion: Without loss of generality we assume that $b$ is in $C_2^a$, $a$ is in $C_4^b$, $c$ is the nearest neighbour of $a$ in $C_2^a$ and $d$ is the nearest neighbour of $b$ in $C_4^b$. In addition, we refer back to this assumption, recalling that $\alpha$ is the clockwise angle $ab$ makes with the vertical axis.
\begin{observation}\label{obs}
	Let $\alpha$ be clockwise angle $ab$ makes with the vertical axis. Then $0\leq \alpha \leq \pi/10$.
\end{observation}

We proceed by induction to bound the spanning ratio of $\Theta_5$.
We show that, for any pair of points \(a, b \in P\),
the length of a shortest path \(|\mathcal{P}(a,b)|\) in \(\Theta_5\)
is at most \(K\) times the Euclidean distance between its
endpoints. The induction is on the rank of the Euclidean distance
\(|ab|\) among all distances between pairs of points in \(P\).
The exact bound on \(K\) is made explicit in the proof.
Lemma~\ref{basecase} is sufficient for the base case of the induction, but we first require the following geometric lemma:



\begin{lemma}\label{trianglelemma}
	Let $\mathcal{T}$ be a triangle $\triangle pqr$,
	and without loss of generality assume that
	$|pq|\leq|pr|$. Then for all points $s \in \mathcal{T}$, $|ps|\leq
	|pr|$.
\end{lemma}

\begin{proof}
	(Figure~\ref{t5:geo})
	Let $s'$ be the intersection of the line through $ps$ onto $qr$, thus
	$|ps| \leq |ps'|$ and it is enough to show that $|ps'| \leq |pr|$.
	The distance from \(p\) to \(s'\) is a convex function of the angle $\angle(spq)$. The minimum of this function is
	attained when the lines through \(ps'\) and \(qr\) are orthogonal.
	Therefore the maximum is attained either at \(s'=q\) or \(s'=r\), whichever
	is furthest.
\end{proof}

\begin{figure}
	\begin{subfigure}[b]{0.45\textwidth}
		\centering
		\includegraphics[page=10,scale=.7]{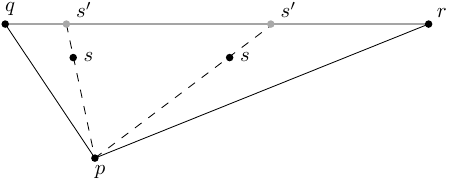}
		\caption{Two examples for the position of $s$.}%
		\label{t5:geo}
	\end{subfigure}%
	\hfill
	\begin{subfigure}[b]{0.45\textwidth}
		\centering
		\includegraphics[page=30,scale=.7]{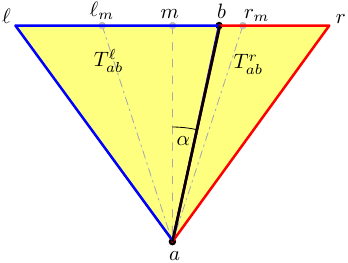}
		\caption{The triangles $T_{ab}^\ell$ and $T_{ab}^r$.}%
		\label{t5:landr}
	\end{subfigure}
	\caption{}
\end{figure}

\state{LemmaBaseCase}{lemma}{\label{basecase}
	Let $(a_0,b_0)$ be the pair of points in $P$ that minimizes \(|ab|\) over
	all points \(a\) and \(b\) in \(P\).
	The $\Theta_5$-graph of $P$ contains the edge $a_0b_0$.%
}

\begin{proof}
	(See Figure~\ref{t5:landr}.)
	Assume by contradiction that \(\Theta_5\) does not contain \(ab\), then
	some point \(p \in P\) different from \(a\) or \(b\) is contained in
	\(T_{ab}\).
	We show that \(|bp| < |ab|\), hence \(ab\) is not the closest pair in
	\(P\).

	Divide $T_{ab}$ into two triangles by separating $T_{ab}$
	along $ab$ into the left triangle $T_{ab}^\ell$ and the right triangle
	$T_{ab}^r$.
	Then \(p\) belongs to one of these triangles.
	Observation \ref{obs} gives us that \(0 \leq \alpha \leq \pi / 10\),
	and thus \(|ba| \geq |b \ell| \geq |br|\) and in both cases we can apply Lemma~\ref{trianglelemma}.
\end{proof}

If \(ab \in \Theta_5\), then $|\pathF{a,b}|\leq K|ab|$ holds for all $K\geq 1$.
Otherwise we assume the following induction hypothesis:
for every pair of points $a', b' \in P$
where $|a'b'|<|ab|$, the shortest path $\pathF{a',b'}$ from $a'$ to
$b'$ has length at most $|\pathF{a',b'}|\leq K\cdot |a'b'|$, for some $K \geq
1$. Our goal is to find the minimum value of $K$ for which our inductive
argument holds.

Recall that $c$ is the closest point to $a$
in $C_2^a$ and $d$ is the closest point to $b$ in $C_4^b$.
We restrict our analysis to the following three paths:
\begin{enumerate}[(1)]
	\item \label{path:unos} $ac+\pathF{c,b}$,
	\item \label{path:dos} $bd+\pathF{d,a}$, and
	\item \label{path:ttres} $ac+\pathF{c,d}+db$.
\end{enumerate}
Depending on the particular arrangement of $a$, $b$, $c$, and $d$, we examine
a subset of these and find a minimum value for $K$ that satisfies at least one
of the following inequalities:
\begin{enumerate}[(A)]
	\item \label{A:unos} $|ac| + K\cdot |cb| \leq K\cdot|ab|$,
	\item \label{B:dos} $|bd| + K\cdot |da| \leq K\cdot|ab|$, and
	\item \label{C:ttres} $|ac| + K\cdot |cd| + |db| \leq K\cdot|ab|$.
\end{enumerate}

Observe that our inductive argument follows if any of these cases holds.
For instance, if we prove \ref{A:unos} holds for some value $K$, it implies
that $|cb|<|ab|$ (since all distances are positive), and thus $|\pathF{c,b}|\leq
K\cdot |cb|$ by the induction hypothesis.
Similar conclusions follow for statements~\ref{B:dos} and~\ref{C:ttres}.
Thus we can combine
\ref{path:unos}-\ref{path:ttres} with \ref{A:unos}-\ref{C:ttres} as follows.
\begin{enumerate}[(a)]
	\item \label{unos} $|\pathF{a,b}| \leq |ac|+|\pathF{c,b}| \leq |ac|+K\cdot |cb| \leq K\cdot|ab|$.
	\item \label{dos} $|\pathF{a,b}| \leq |bd|+|\pathF{d,a}| \leq |bd|+K\cdot |da| \leq K\cdot|ab|$.
	\item \label{ttres} $|\pathF{a,b}| \leq |ac|+|\pathF{c,d}|+|db| \leq |ac|+K\cdot |cd|+|db| \leq K\cdot|ab|$.
\end{enumerate}

For any given arrangement of vertices we prove that at least one of
\ref{A:unos}, \ref{B:dos}, or \ref{C:ttres} holds true for some value $K$, and
find the smallest value for which this is true. Our proof relies mainly on case
analysis, but some of these cases have similar structure.
We exploit this structure in Section~\ref{casebycase}
by designing two geometric lemmas that we apply repeatedly in the inductive step.
These lemmas, along with additional arguments, are then applied to
different arrangements of $a$, $b$, $c$, and $d$.
For all but one case we show that at least one of~\ref{unos},~\ref{dos},
or~\ref{ttres} holds true for $K\geq 5.70$.
The last case requires $K\geq 6.16$.
We improve this further to $K\geq 5.70$,
but due to the complexity of this last case,
we dedicate Section~\ref{down2K} to its analysis.

\section{Analysis}\label{casebycase}

We first introduce two triangles $\Tr$ and $\Tc$ for which inequalities of the form of~\ref{A:unos}
and~\ref{B:dos} hold for reasonable values of \(K\) (see Figure~\ref{examplet5s:sketch}). Note the triangles are numbered to correspond to the lemmas they appear in. 
We state these inequalities as lemmas whose
repeated use simplifies the proof of our main result. 
%


\begin{figure}
	\centering
	\begin{subfigure}[b]{0.45\textwidth}
		\includegraphics[page=1,scale=.7]{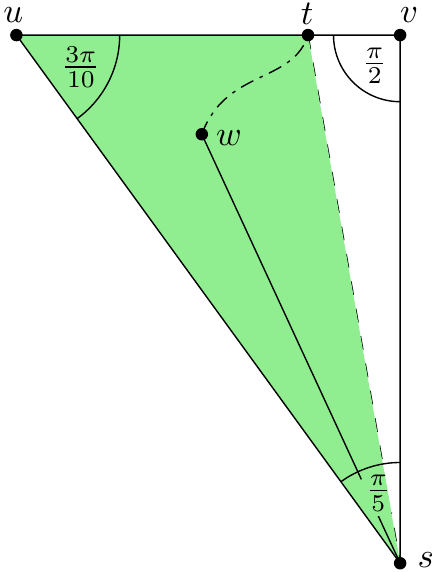}
		\caption{%
			$\Tr$ has angles $(\frac{\pi}{5},\frac{\pi}{2},\frac{3\pi}{10})$.%
		}%
		\label{fig:t253:sketch}
	\end{subfigure}%
	\hfill
	\begin{subfigure}[b]{0.45\textwidth}
		\centering
		\includegraphics[page=2,scale=.7]{figures/generic-triangles-sketch.pdf}
		\caption{
			$\Tc$ has angles $(\frac{3\pi}{10},\frac{3\pi}{10},\frac{2\pi}{5})$.%
		}%
		\label{fig:t334:sketch}
	\end{subfigure}
	\caption{Triangles $\Tr$ and $\Tc$.}%
	\label{examplet5s:sketch}
\end{figure}



\state{LemmaTRight}{lemma}{\label{t253}
	(Figure~\ref{fig:t253:sketch})
	Let $\Tr$ be a triangle with vertices \((s,v,u)\) and corresponding
	interior angles $(\frac{\pi}{5},\frac{\pi}{2},\frac{3\pi}{10})$.
	Let \(t\) be a point on \(uv\) and
	let \(w\) be a point inside \(\triangle stu\).
	Then $|sw| + K|wt| \leq K|st|$ for all $K \geq 4.53$.%
}

\begin{proof}
	(Figure~\ref{examplet52})
	We show $\Phi = |\qs\qw|+K|\qw\qt|-K|\qs\qt| \leq 0$.
	Without loss of generality, orient $\Tr$ so that $u$ and $v$ define a
	horizontal line with \(u\) left of \(v\) and with \(s\) below that line.
	Let $\qw_r$ be the horizontal
	projection of $\qw$ onto $\qs\qt$, and let $\qw_\ell$ be the horizontal
	projection of $\qw$ onto $\qs\qu$. We have
	$|\qw\qw_r|+|\qw_r\qt|\geq|\qw\qt|$ by the triangle inequality. We also
	have that $\angle \qs\qw\qw_\ell \geq \pi/2$, which implies that
	$\qs\qw_\ell$ is the longest edge in triangle $\qs\qw\qw_\ell$
	(the triangle can be drawn inside a disk whose diameter is \(sw_{\ell}\)),
	and thus
	$|\qs\qw_\ell|\geq |\qs\qw|$. Since $\qw$ is on $\qw_\ell\qw_r$, we have
	$|\qw_\ell\qw_r|\geq |\qw\qw_r|$. Thus
	\begin{align*}
	\Phi &= |\qs\qw|+K|\qw\qt|- K|\qs\qt|\\
	&\leq |\qs\qw_\ell|+K(|\qw\qw_r|+|\qw_r\qt|) -K(|\qs\qw_r|+|\qw_r\qt|)\\
	&\leq |\qs\qw_\ell|+K|\qw_\ell\qw_r|- K|\qs\qw_r| = \Phi'.
	\end{align*}

	Let \(\beta = \angle vst \geq 0\).
	Observe that \(\Phi'\) increases as $\beta$ decreases,
	since $|\qs\qw_r|$ decreases while
	$|\qw_\ell\qw_r|$ increases and $|\qs\qw_\ell|$ stays constant.
	Hence,
	$\Phi'$ is maximized when $\beta = 0$, that is, when $\qw_r$ lies on
	$\qs\qv$.
	Thus assume that $\qw_r$ lies on $\qs\qv$ and let $|\qs\qw_\ell|=1$ without
	loss of generality.
	We bound \(\Phi'\) in terms of \(\angle w_rsw_{\ell} = \frac \pi5\):
	\begin{displaymath}
		\Phi'
		\leq
		1 + K \sin\left(\frac \pi5\right) - K \cos\left(\frac \pi5\right).
	\end{displaymath}

	Solving for \(K\) we get \(\Phi \leq \Phi' \leq 0\) when
	\begin{displaymath}
		K
		\geq \frac{1}{
			\cos(\frac \pi5)-\sin(\frac \pi5)
		}
		= 4.52\ldots
	\end{displaymath}
\end{proof}

\begin{figure}
	\centering
	\begin{subfigure}[b]{0.45\textwidth}
		\includegraphics[page=1,scale=.7]{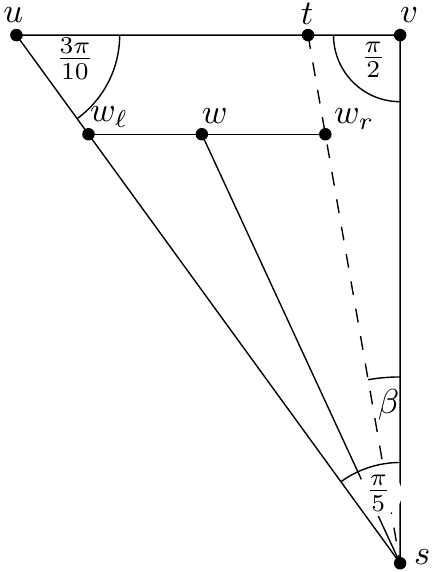}
		\caption{$\Tr$ has angles $(\frac{\pi}{5},\frac{\pi}{2},\frac{3\pi}{10})$.}%
		\label{examplet52}
	\end{subfigure}%
	\hfill
	\begin{subfigure}[b]{0.45\textwidth}
		\centering
		\includegraphics[page=2,scale=.7]{figures/generic-triangles.pdf}
		\caption{$\Tc$ has angles $(\frac{3\pi}{10},\frac{3\pi}{10},\frac{2\pi}{5})$.}%
		\label{examplet53}
	\end{subfigure}
	\caption{Detailed analysis of triangles $\Tr$ and $\Tc$.}%
	\label{examplet5s}
\end{figure}


\state{LemmaTCone}{lemma}{\label{t334}
	(Figure~\ref{fig:t334:sketch})
	Let $\Tc$ be a triangle with vertices \((s,v,u)\) and corresponding interior angles
	$(\frac{3\pi}{10},\frac{3\pi}{10},\frac{2\pi}{5})$.
	Let $t$ be a point on
	$uv$ such that $\angle vst \leq \pi/10$
	and
	let \(w\) be a point inside \(\triangle stu\).
	Then $|sw| + K|wt| \leq K|st|$ for all $K \geq 5.70$.%
}

\begin{proof}
	(Figure~\ref{examplet53})
	We show $\Phi = |\qs\qw|+K|\qw\qt|-K|\qs\qt| \leq 0$ by case analysis.

	Case 1) $\angle vsw \leq \frac \pi5$ (Figure~\ref{useTr}):
	Let $v'$ be the orthogonal
	projection of $\qt$ onto $\qs\qv$. Let $\qu'$ be the point on the line
	through $\qt$ and $v'$ such that $\angle v'\qs\qu' = \frac \pi5$. Observe
	that $\triangle sv'u'$ corresponds to $\Tr$ of Lemma \ref{t253}
	and it contains $\qw$. Thus Lemma \ref{t253}
	tells us $\Phi \leq 0$ for all $K \geq 4.53$.

	Case 2) $\angle vsw > \frac \pi5$ (Figure~\ref{Tc}):
	Without loss of generality, orient $\Tc$ so that $u$ and $v$ define a
	horizontal line with \(u\) left of \(v\) and with \(s\) below that line.
	Let $\qw_r$ be the horizontal projection of $\qw$ onto $\qs\qt$, and let
	$\qw_\ell$ be the horizontal projection of $\qw$ onto $\qs\qu$.
	We have
	$|\qw\qw_r|+|\qw_r\qt|\geq|\qw\qt|$ by the triangle inequality.
	We also have
	that $\angle \qs\qw\qw_\ell > \pi/2$, which implies that $\qs\qw_\ell$
	is the longest edge in $\triangle \qs\qw\qw_\ell$ (the triangle can be
	drawn inside a disk whose diameter is \(sw_{\ell}\)), and thus
	$|\qs\qw_\ell|\geq |\qs\qw|$. Since $\qw$ is on $\qw_\ell\qw_r$, we have
	$|\qw_\ell\qw_r|\geq |\qw\qw_r|$. Thus
	\begin{align*}
	\Phi &= |\qs\qw|+K|\qw\qt|- K|\qs\qt|\\
	&\leq |\qs\qw_\ell|+K(|\qw\qw_r|+|\qw_r\qt|) -K(|\qs\qw_r|+|\qw_r\qt|)\\
	&\leq |\qs\qw_\ell|+K|\qw_\ell\qw_r|- K|\qs\qw_r| = \Phi'.
	\end{align*}
    We rewrite \(\Phi'\) in terms of \(\beta = \angle vst \geq 0\) using the
	sine law we get
    $$
    |\qs\qw_\ell| = \frac{|\qs\qw_r|\sin\left(\frac{3\pi}{10}+\beta\right)}{\sin\left(\frac{2\pi}{5}\right)}
    $$
    and 
    $$
    |\qw_\ell\qw_r| = \frac{|\qs\qw_r|\sin\left(\frac{3\pi}{10}-\beta\right)}{\sin\left(\frac{2\pi}{5}\right)}.
    $$
    We normalize \(\Phi'\) by dividing each term by $\frac{|\qs\qw_r|}{\sin\left(\frac{2\pi}{5}\right)}$ which gives us
	$$
	\Phi' = \sin\left(\frac{3\pi}{10}+\beta\right) + K
	\sin\left(\frac{3\pi}{10}-\beta\right) -K
	\sin\left(\frac{2\pi}{5}\right).
	$$
	The derivative of \(\Phi'\) with respect to \(\beta\) is
	$$
	\frac{d \Phi'}{d \beta} = \cos\left(\frac{3\pi}{10}+\beta\right) - K\cos
	\left(\frac{3\pi}{10}-\beta\right).
	$$
	For all \(K \geq 1\),
	\(\frac{d \Phi'}{d \beta}(0)\) is negative
	and
	\(\frac{d \Phi'}{d \beta}(\beta)\)
	is monotone decreasing for \(0 \leq \beta \leq \frac{\pi}{10}\).
	Hence \(\frac{d \Phi'}{d \beta}\) is negative on the whole range \((K \geq 1)
	\times (0 \leq \beta \leq \frac{\pi}{10})\)
	and
	\(\Phi'\) is maximized at \(\beta = 0\) for all \(K \geq 1\). Thus
	\begin{displaymath}
		\Phi'
		\leq \Phi'(0)
		=
		\sin\left(\frac{3\pi}{10}\right)
		+ K \sin\left(\frac{3\pi}{10}\right)
		- K \sin\left(\frac{2\pi}{5}\right).
	\end{displaymath}
	Solving for \(K\) we get $\Phi \leq \Phi' \leq 0$ when
	\begin{displaymath}
		K
		\geq \frac{
			\sin\left(\frac{3\pi}{10}\right)
		}{
			\sin\left(\frac{2\pi}{5}\right)-\sin\left(\frac{3\pi}{10}\right)
		}
		= 5.69\ldots
	\end{displaymath}
\end{proof}

\begin{figure}
	\centering
	\begin{subfigure}[b]{0.45\textwidth}
		\centering
		\includegraphics[page=4,scale=.7]{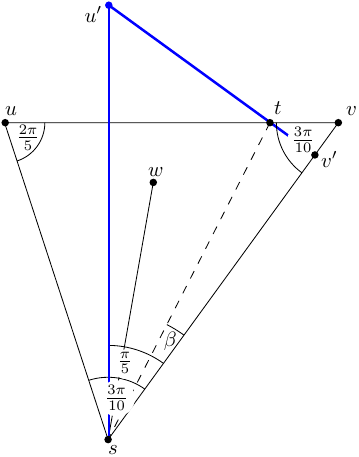}
		\caption{Using $\Tr$ to analyze $\Tc$.}%
		\label{useTr}
	\end{subfigure}%
	\hfill
	\begin{subfigure}[b]{0.45\textwidth}
		\centering
		\includegraphics[page=3,scale=.7]{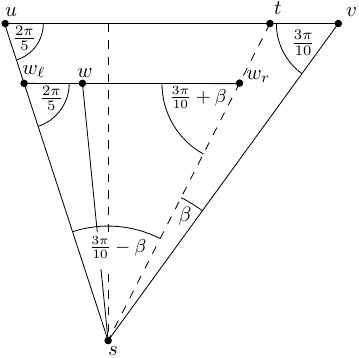}
		\caption{Proving $|\qs\qw_\ell|+K|\qw_\ell\qw_r|\leq K|\qs\qw_r|$}%
		\label{Tc}
	\end{subfigure}
	\caption{Analyzing triangle $\Tc$.}
\end{figure}

As in the definition of \(T_{ab}\) and \(T_{ba}\) in Section \ref{prelim}, let \(c\) be the point closest
to \(a\) in \(T_{ab}\) and let \(d\) be the point closest to \(b\) in
\(T_{ba}\).
We proceed by case analysis depending on the location of the points
\(c\) and \(d\).

\aurelien{Explicit the case \(d = c\)?}

If $c$ is to the right of \(ab\) or if $d$ is to the right of $ab$,
we can apply Lemma~\ref{t253} to show the existence of a short
path from $a$ to $b$.
When both $c$ and $d$ are left of $ab$,
we use a more complicated argument requiring a new definition:

\begin{definition}\label{def:pentagon}
	(Figure~\ref{p5})
	Given any pair of points \((a,b)\) in \(P\), let \(r'\) and \(r'_m\) be as
	in the definition of \(T_{ba}\) in Section \ref{prelim}.
	We define $P_{ab}$ to be the regular pentagon with vertices
	$(p_0,p_1,p_2=r',p_3=r'_m,p_4)$
	where \(p_4\) is above the line going through \(r'\) and \(r'_m\)
	(this uniquely defines the remaining points \(p_0\) and \(p_1\)).
\end{definition}

\begin{figure}
	\centering
	\begin{subfigure}[b]{0.45\textwidth}
		\centering
		\includegraphics[page = 20,scale=.7]{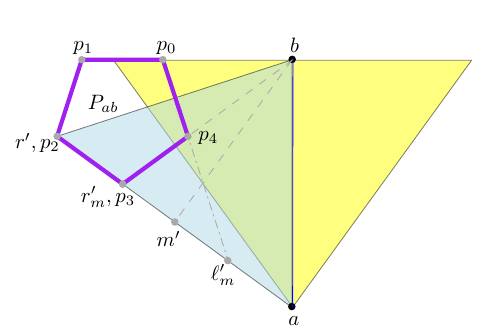}
		\caption{$P_{ab}$ when $\alpha=0$.}%
		\label{p51}
	\end{subfigure}%
	\hfill
	\begin{subfigure}[b]{0.45\textwidth}
		\centering
		\includegraphics[page = 21,scale=.7]{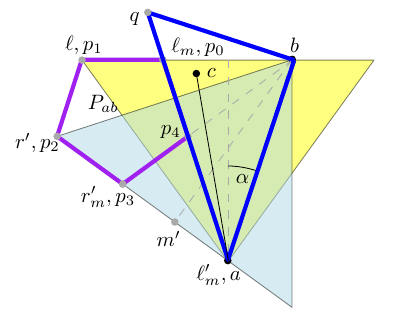}
		\caption{$P_{ab}$ when $\alpha = \pi/10$.}%
		\label{p52}
	\end{subfigure}
	\caption{The regular pentagon \(P_{ab}\).}%
	\label{p5}
\end{figure}

Observe that $P_{ab}$ is fixed with respect to $T_{ba}$.
This construction puts \(p_4\) inside \(T_{ab}\) and
puts $p_0$ and $p_1$ on a horizontal line with $b$, with
\(p_0\) lying on the boundary of \(T_{ab}\).

\state{ClaimPentagon}{claim}{\label{claim:pentagon}
	Given Definition~\ref{def:pentagon} we have that
	\(p_4 \in T_{ab}\),
	\(p_0 \in \ell b\),
	and \(p_1\) lies on the line through \(\ell\) and \(b\).%
}

\begin{wrapfigure}[12]{l}{0.4\textwidth}
	\centering
	\includegraphics[page=14,width=.35\textwidth]{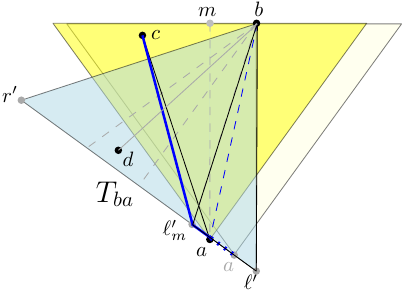}
	\caption{Transformation~\ref{transformt5}.}%
	\label{ctransform}
\end{wrapfigure}

\begin{proof}
	Note that $p_3p_4$ and $p_3b$ share the same supporting line
	since \(\angle p_2 p_3 p_4 = \angle p_2 p_3 b = \frac{3 \pi}{5}\).
	Let $f$ be the intersection of $a\ell$ and $p_3b$.
	Given this observation and this definition,
	it is equivalent to prove that \(p_4\) lies in the segment \(fb\).

	 Translate \(a\) on the segment \(\ell'_m \ell'\).
	Since the slope of \(a \ell\) is smaller than the slope of \(\ell'_m
	\ell'\), translating \(a\) to \(a = \ell'\), that is letting \(\alpha =
	0\), maximizes the \(y\)-intercept of the line going through \(a\) and
	\(\ell\) with any fixed vertical line.
	Hence this translation shrinks \(fb\), and
	it remains to prove that \(p_4\) stays in \(fb\)
	only in that extreme case.

	With the simplifying assumption that \(\alpha = 0\),
	we show that $|p_3f| < |p_3p_4| < |p_3b|$, which proves the claim.
	Note that $\angle \ell a p_3 = \pi/10$ and \(\angle p_3 f a = \pi/2\),
	thus $|p_3f| = |p_3a|\sin(\pi/10)$.
	We have $|p_3p_4| = |p_3p_2|
	= |p_3a|\sin{(\frac{\pi}{10}})/\sin{(\frac{3 \pi}{10})}$.
	Since $|p_3a|=|p_3b|$, we obtain
	\begin{displaymath}
		|p_3b|\sin{\left(\frac{\pi}{10}\right)}
		<
		|p_3b|
		\frac{\sin{(\frac{\pi}{10}})}{\sin{(\frac{3 \pi}{10})}}
		<
		|p_3b|.
	\end{displaymath}
\end{proof}

Given this definition, we consider the following cases:
When \(c\) is not in \(P_{ab}\) we prove $|ac| + |\pathF{c,b}| \leq 5.70 |ab|$.
When \(d\) is not in \(P_{ab}\) we prove $|bd| + |\pathF{d,a}| \leq 5.70 |ab|$.
When both $c$ and $d$ are in $P_{ab}$ we analyze the
length of the path $ac + \pathF{c,d} + db$.
Lemma~\ref{case8} gives us a bound of $6.16 |ab|$ with a simple proof.
Using a more technical analysis, we obtain a bound of $5.70 |ab|$.
This is proven in Lemma~\ref{newandimproved} in Section~\ref{down2K}.

Some of the proofs use the simplifying assumption that \(\alpha = \pi/10\).
This is achieved through the following transformation:
given \(a\), \(b\), \(c\), \(d \in P\) with \(T_{ab}\) and
\(T_{ba}\) as defined earlier, we define:
\begin{transformation}\label{transformt5}
	Fix $b$, $c$, $d$, and $T_{ba}$, and translate $a$ along $r'\ell'$ until $a = \ell'm$.
\end{transformation}

See Fig. \ref{ctransform}. Observe that this transformation changes $|ac|$ and $|ab|$, but not
$|bd|$, $|cd|$, or $|cb|$. The transformation also changes $|ad|$,
but we do not use it in any case that depends on this value.
We prove the following lemma allowing the application of Transformation~\ref{transformt5}
without loss of generality in several cases.

\state{LemmaTransformation}{lemma}{\label{posofc}
	Under Transformation~\ref{transformt5},
	the values of $|bd|$, $|cd|$, and $|cb|$ are unchanged, and
	$\Psi = |ac| - K|ab|$ is maximized when $a = \ell_m'$ for all
	\(K \geq 3.24\).%
}

\begin{wrapfigure}[1]{r}{0.20\textwidth}
	\centering
	\includegraphics[page=1]{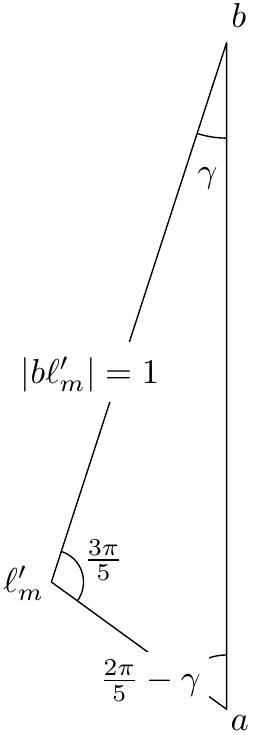}
	\caption{The values used in the proof of Lemma~\ref{posofc}.}%
	\label{ctransform-lemma-details}
\end{wrapfigure}
\noindent
\begin{minipage}[t]{0.75\textwidth}
	\begin{proof}
	(Figure~\ref{ctransform-lemma-details})
	Let $\gamma = \angle\ell_m'ba = \pi/10-\alpha$.
	Define $\Psi' = |a\ell_m'|+ |\ell_m'c| - K|ab|.$ Note by the triangle inequality that $\Psi'\geq \Psi$.
	We show that $\Psi'$ is monotonically decreasing in $\gamma$, which proves
	both $\Psi$ and $\Psi'$ are maximized when $\gamma = 0$ since then \(\Psi =
	\Psi'\).
	We let \(| b \ell_{m}' | = 1\) without loss of generality
	and express \(\Psi'\) as a function of \(\gamma\)
	using the law of sines:

	Using
	$
		|a\ell_m'| = \frac{\sin{\gamma}}{\sin{(\frac{2 \pi}{5} - \gamma)}}
	$
	and
	$
		|ab| = \frac{\sin{(\frac{3 \pi}{5})}}{\sin{(\frac{2 \pi}{5} - \gamma)}}
			= \frac{\sin{(\frac{2 \pi}{5})}}{\sin{(\frac{2 \pi}{5} -
			\gamma)}}
	$,
	we have
	\begin{displaymath}
		\Psi' = \frac{\sin{\gamma} - K \sin{(\frac{2 \pi}{5})}}{\sin{(\frac{2
		\pi}{5} - \gamma)}} + \underbrace{| \ell_m' c
	|}_{\text{Independent of }\gamma}.
	\end{displaymath}
	Hence,
	\begin{align*}
	\frac{d\Psi'}{d\gamma}
	&= \frac{\cos\gamma\sin(\frac{2\pi}{5}-\gamma) +
	\cos(\frac{2\pi}{5}-\gamma)(\sin\gamma -
K\sin(\frac{2\pi}{5}))}{\sin^2(\frac{2\pi}{5}-\gamma)}\\
	&= \frac{\sin(\frac{2\pi}{5}) ( 1 - K
	\cos(\frac{2\pi}{5}-\gamma))}{\sin^2(\frac{2\pi}{5}-\gamma)}.
	\end{align*}
	Since \(0 \leq \gamma \leq \pi/10\),
	the denominator is positive on the whole range
	and the numerator is maximized when $\gamma = 0$.
	Since $\sin(\frac{2\pi}{5})$ is positive, it suffices to satisfy
	$1 - K \cos(\frac{2\pi}{5}) \leq 0$:
	\begin{displaymath}
		K \geq \frac{1}{\cos{(\frac{2 \pi}{5})}} = 3.23\ldots
	\end{displaymath}
\end{proof}

\end{minipage}
\vspace{0.5cm}
\\By Lemma \ref{posofc} we see that by applying Transformation \ref{transformt5} we maximize the value $|ac|-K|ab|$. Another way to see this is that we minimize $K|ab|$. This, in turn, allows us to explicitly determine under what conditions the inductive hypothesis applies. Note that applying Transformation~\ref{transformt5} to where \(a = \ell_m'\)
is equivalent to assuming $\alpha = \pi/10$.

All these proofs can be combined in an analysis comprising
\emph{eight} cases
depending on the location of \(c\) and \(d\) with respect to \(T_{ab}\),
\(T_{ba}\), and \(P_{ab}\), as illustrated below in the breakdown of the case analysis below. In each case we prove that for a given  arrangement of vertices that $|\pathF{a,b}|\leq K|ab|$ for the given value $K$.\\

\begin{algorithm}[H]
	\renewcommand*{\algorithmcfname}{Breakdown of the case analysis}
\SetAlgoLined
\vspace{0.25cm}
\begin{enumerate}
	\item If \(c\) is right of \(ab\), then $K\geq 4.53$ by Lemma~\ref{case1}.
	\item If \(d\) is right of \(ab\), then $K\geq 4.53$ by Lemma~\ref{case4}.
	\item Else both $c$ and $d$ are left of \(ab\). We have the following cases:
	\begin{enumerate}
		\item If $c$ is in $T_{ba}$, then $K\geq 5.70$ by Lemma~\ref{case2}.
		\item Else $c$ is NOT in $T_{ba}$ and:
		\begin{enumerate}
			\item If $c$ is NOT in $P_{ab}$ then $K\geq 4.53$ by Lemma~\ref{case3}.
			\item Else $c$ is in $P_{ab}$ and:
			\begin{itemize}
					\item If $d$ is right of $am$ then $K\geq 3.24$ by Lemma~\ref{case5}.
					\item If $d$ is left of $am$ and above $c$ then $K\geq 4.53$ by  Lemma~\ref{case6}
					\item If \(d\) is left of $am$ and below \(c\) (i.e. \(d \not\in T_{ab}\) such that $bd$ and $ac$ cross)
					\begin{itemize}
						\item If $d$ is NOT in $P_{ab}$ then $K\geq 5.70$ by  Lemma~\ref{case7}.
						\item If $d$ is in $P_{ab}$ then \(K \geq 6.16\) by  Lemma~\ref{case8} 
						 or \(K \geq 5.70\) by Lemma~\ref{newandimproved}.
					\end{itemize}
				\end{itemize}
		\end{enumerate}
	\end{enumerate}
\end{enumerate}
\caption{}
\end{algorithm}
\vspace{0.5cm}
One can check that all locations of \(c\) and \(d\) are covered.
This proves our main theorem:
\state{TheoremMain}{theorem}{\label{main}
	Given a set \(P\) of points in the plane,
	the \(\Theta_5\)-graph of \(P\) is a \(5.70\)-spanner.
}

We use the remainder of the paper to prove each lemma.

\begin{figure}
	\centering
	\begin{minipage}{0.45\textwidth}
		\centering
		\includegraphics[page = 4,scale=.9]{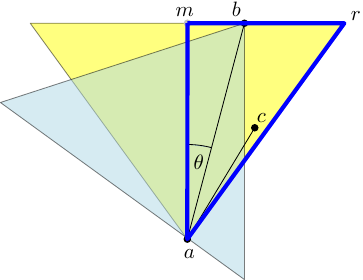}
		\caption{Points $(a,r,m)$ correspond to $\Tr$ (in blue)
		with $\qt = b$ and $\qw = c$.}%
		\label{crightab}
	\end{minipage}%
	\hfill
	\begin{minipage}{0.45\textwidth}
		\centering
		\includegraphics[page = 7,scale=.9]{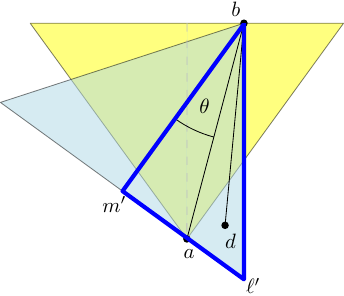}
		\caption{Points
			$(b,m',\ell')$ correspond to $\Tr$ (in blue) with $\qt=a$ and
			$\qw=d$.}%
		\label{fig:t253}
	\end{minipage}
\end{figure}
\newpage
\begin{lemma}\label{case1}
	If $c$ is right of $ab$, then $|\pathF{a,b}|\leq K|ab|$ for $K\geq 4.53$.
\end{lemma}

\begin{proof}
	(Figures~\ref{examplet5s:sketch},~\ref{crightab})
	Let $(\qs,\qt,\qw, \qu,\qv) = (a,b,c, r, m)$, thus
	these points correspond to triangle $\Tr$ of Lemma~\ref{t253}.
	Thus $|ac|+K|cb|\leq K|ab|$ for all $K\geq 4.53$.
	The induction hypothesis and Lemma~\ref{t253} imply that there is a path
	from $a$ to $b$ with length at most
	\begin{displaymath}
	|\pathF{a,b}| \leq |ac|+|\pathF{c,b}|\leq |ac| + K|cb|\leq K|ab|.
	\end{displaymath}
\end{proof}

\begin{lemma}\label{case4}
	If $d$ is right of $ab$, then $|\pathF{a,b}|\leq K|ab|$ for $K\geq 4.53$.
\end{lemma}

\begin{proof}
	(Figures~\ref{examplet5s:sketch},~\ref{fig:t253})
	Let $(\qs,\qt,\qw,\qu,\qv) = (b,a, d, m',\ell')$, thus these points
	correspond to triangle $\Tr$ from Lemma~\ref{t253}.
	Thus $|bd| + K|da| \leq K|ab|$ for $K\geq
	4.53$ by Lemma~\ref{t253}.
	The induction hypothesis and Lemma~\ref{t253} imply that there
	is a path from $a$ to $b$ with length at most
	\begin{displaymath}
	|\pathF{a,b}| \leq |bd|+|\pathF{d,a}|\leq |bd| + K|da|\leq K|ab|. 
	\end{displaymath}
\end{proof}

\begin{lemma}\label{case2}
	If $c$ is left of $ab$ and in $T_{ab}\cap T_{ba}$, then $|\pathF{a,b}|\leq K|ab|$ for $K\geq 5.70$.
\end{lemma}

\begin{figure}
	\centering
	\begin{minipage}{0.45\textwidth}
		\centering
		\includegraphics[page = 6,scale=.9]{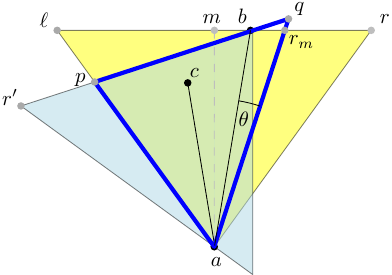}
		\caption{Points $(a,q,p)$ correspond to the triangle $\Tc$ with angles
		$(\frac{3\pi}{10},\frac{2\pi}{5},\frac{3\pi}{10})$ as denoted by the
		blue triangle. Let $\qt = b$ and $\qw = c$, and $\theta =
		\frac{\pi}{10}-\alpha$, which falls in the range of $0\leq \angle \qv\qs\qu
		\leq \pi/10$.}%
		\label{cintersect}
	\end{minipage}%
	\hfill
	\begin{minipage}{0.45\textwidth}
		\centering
		\includegraphics[page=11,scale=.7]{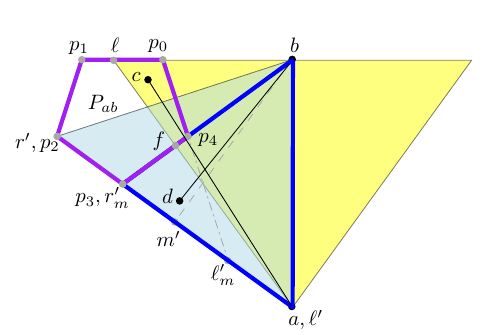}
		\caption{We use the fact that $p_4$ lies in $T_{ab}$ and apply $\Tc$.}%
		\label{examplet59}
	\end{minipage}
\end{figure}

\begin{proof}
	(Figures~\ref{examplet5s:sketch},~\ref{cintersect})
	Let $p$ be the intersection of $br'$ and $a\ell$, and let $q$ be the
	intersection of the lines through $r'b$ and $ar_m$. Observe that $0\leq
	\angle r_mab \leq \pi/10$, thus $\angle r_mab$ has the same range as
	$\angle \qv\qs\qt$ from $\Tc$ in Lemma \ref{t334}. If we let points
	$(\qs,\qt,\qw,\qu, \qv) = (a,b,c,p,q)$, then these points correspond to the
	triangle $\Tc$, and thus $|ac|+K|cb|\leq K|ab|$ for $K\geq 5.70$ by Lemma
	\ref{t334}. Our induction hypothesis and Lemma \ref{t334}
	imply that there is a path from $a$ to $b$ with length
	\begin{displaymath}
	|\pathF{a,b}|\leq |ac|+|\pathF{c,b}|\leq |ac| + K|cb|\leq K|ab|. 
	\end{displaymath}
\end{proof}

\begin{lemma}\label{case3}
	If $c \in T_{ab} \setminus (T_{ba} \cup P_{ab})$,
	then $|\pathF{a,b}|\leq K|ab|$ for all $K \geq 4.53$.
\end{lemma}

\begin{proof}
	(Figures~\ref{examplet5s:sketch},~\ref{p52})
	Let $\Phi = |ac|+K|cb|- K|ab|$.
	We apply Transformation \ref{transformt5}. Since $c \not\in T_{ba}$ it
	must be left of $\ell_m'b$, thus $c$ remains left of $ab$. As $a$ moves
	left along $\ell'\ell'_m$, so does the left side of $T_{ab}$, which means that $c$ remains
	inside $T_{ab}$. Thus Lemma \ref{posofc} implies that $\Phi$ is maximized
	at $\alpha = \pi/10$, thus we assume this is the case. Observe that $\angle
	ba\ell_m = \pi/5$, and $\angle \ell_mba = 2\pi/5 <\pi/2$. Let $q$ be the
	intersection of the line through $b$ orthogonal to $ab$ and the line
	through $a$ and $\ell_m$. If we let $(\qs,\qt,\qw,\qu,\qv) = (a,b,c,q,b)$
	then these points correspond to $\Tr$. Then Lemma
	\ref{t253} tells us that $|ac|+K|cb|\leq K|ab|$ and thus
	$\Phi=|ac|+K|cb|- K|ab| \leq 0$ for all $K \geq 4.53$.
\end{proof}

\begin{figure}
	\centering
	\begin{minipage}{0.45\textwidth}
		\centering
		\includegraphics[page=12,scale=.7]{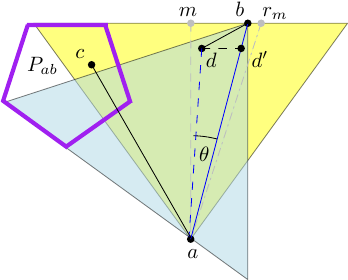}
		\caption{The point $c$ is in $P_{ab} \setminus T_{ba}$, and $d$ is right of $am$.}%
		\label{examplet510}
	\end{minipage}%
	\hfill
	\begin{minipage}{0.45\textwidth}
		\centering
		\includegraphics[page=13,scale=.7]{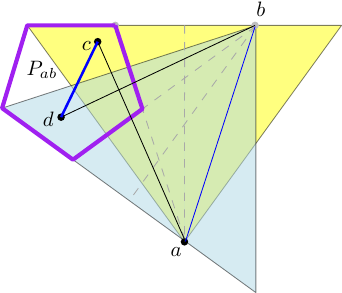}
		\caption{The segments $ac$ and $bd$ cross and \(c\) and \(d\) are in $P_{ab}$.}%
		\label{acbdcross}
	\end{minipage}
\end{figure}

\begin{lemma}\label{case5}
	If $d$ is left of $ab$ and right of $am$, then $|\pathF{a,b}|\leq K|ab|$
	for $K\geq 3.24$.
\end{lemma}

\begin{proof}
	(Figure~\ref{examplet510})
	We show $\Phi = |bd| + K|da| - K|ab| \leq 0$, which implies
	$|\pathF{a,b}| \leq  |bd| + |\pathF{d,a}| \leq K |ab|$ by the triangle
	inequality and the induction hypothesis.

	Let $d'$ be the horizontal projection of $d$ onto $ab$. Let $\Phi_1
	= |bd|-K|bd'|$ and $\Phi_2 = K|da|-K|d'a|$, and note that $\Phi =
	\Phi_1+\Phi_2$ since \(d' \in ab\). Thus it is sufficient to show that $\Phi_1\leq0$ and
	$\Phi_2\leq 0$.

	Observe that $\angle d'da >\pi/2$, since $d$ is right of
	$am$, thus $|d'a|>|da|$, and $\Phi_2\leq0$ for all $K\geq 1$.
	For $\Phi_1\leq 0$ we need $K\geq \frac{|bd|}{|bd'|}$. Let $\gamma =\angle d'db$ and note that $\gamma \geq \pi/10$ because
	\(d \in T_{ba}\). Let $d_y(b,d')$ be the vertical distance between $b$ and $d'$. We have $\sin\gamma = \frac{d_y(b,d')}{|bd|}$. Observe that  $d_y(b,d')\leq |bd'|$ and thus $\frac{|bd|}{|bd'|}\leq \frac{|bd|}{d_y(b,d')}=\frac{1}{\sin\gamma}\leq \frac{1}{\sin(\pi/10)}$.
	Thus $K\geq\frac{1}{\sin(\pi/10)}\geq  \frac{|bd|}{d_y(b,d')}$, and $K\geq \frac{1}{\sin(\pi/10)} = 3.23\dots$
	is sufficient.
\end{proof}
\begin{figure}
	\begin{subfigure}[b]{0.5\textwidth}
		\centering
		\includegraphics[page=16,scale=.7]{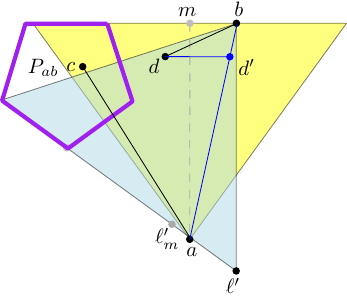}
		\caption{We have $|db|-K|bd'|\leq0$.}%
		\label{dleftam}
	\end{subfigure}%
	\begin{subfigure}[b]{0.5\textwidth}
		\centering
		\includegraphics[page=17,scale=.7]{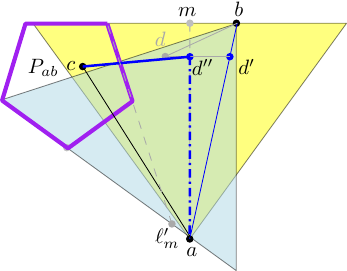}
		\caption{We have $|ac|+K|cd|-K|ad'|\leq0$.}%
		\label{dleftam2}
	\end{subfigure}
	\caption{The point $c$ is in $P_{ab} \setminus T_{ba}$, and $d$ is left of $am$ but above $c$.}%
	\label{examplet511}
\end{figure}

\begin{lemma}\label{case6}
	If $c$ is in $P_{ab} \setminus T_{ba}$, and $d$ is left of $am$ but above
	$c$, then $|\pathF{a,b}|\leq K|ab|$ for all $K \geq 4.53$.
\end{lemma}

\begin{proof}
	(Figures~\ref{examplet5s:sketch},~\ref{examplet511})
	We show $\Phi = |ac| + K|cd| + |db| - K|ab| \leq 0$, which implies
	$|\pathF{a,b}| \leq |ac| + |\pathF{c,d}| + |db| \leq K |ab|$ by the triangle
	inequality and the induction hypothesis.
	We split $\Phi$ into two parts, and show that each part is less than $0$.
	Let $d'$ be the horizontal projection of $d$ onto $ab$. Let $\Phi_1 = |bd|
	- K|bd'|$, and let $\Phi_2 = |ac|+K|cd|-K|ad'|$. Observe that $\Phi =
	\Phi_1+\Phi_2$ since \(d' \in ab\).

	To show that $\Phi_1 \leq 0$,
	observe that $d_y(b,d) = d_y(b,d') \leq |bd'|$.
	Thus let $\Phi_1' = |bd| - K\cdot d_y(b,d) \geq \Phi_1$.
	Let $\gamma = \angle d'db$, and observe that $\Phi_1' =
	|bd|(1-K\sin\gamma)$.  Note that $\gamma \geq \pi/10$ since \(d \in
	T_{ba}\), and thus $K \geq 3.24$ is sufficient to have  $\Phi_1 \leq 0$.

	For $\Phi_2 \leq 0$,
	let $d''$ be the horizontal projection of $d$ onto $am$. Since
	$\angle ad''d' = \pi/2$, $|ad''|\leq |ad'|$. Since $c \not\in T_{ba}$,
	$\angle cdd'' \geq 9\pi/10$, thus $|cd''|>|cd|$. Let $\Phi_2' =
	|ac|+K|cd''|-K|ad''| \geq \Phi_2$. Let $q$ be
	the horizontal projection of $d''$ onto $a\ell$.
	Let the points $(\qs,\qt,\qw,\qu,\qv) = (a,d'',c,q,d'')$ and thus these
	points correspond to $\Tr$. Thus $|ac|+K|cd''|\leq K|ad''|$
	for all \(K \geq 4.53\) by Lemma~\ref{t253}. 
\end{proof}

\begin{lemma}\label{case7}
	If $d$ is left of \(ab\), below $c$ and not in $P_{ab}$,
	then $|\pathF{a,b}|\leq K|ab|$ for all $K\geq5.70$.
\end{lemma}

\begin{proof}
	(Figures~\ref{examplet5s:sketch},~\ref{examplet59})
	We note that $ac$ and $bd$ intersect and $d$ must be outside of $T_{ab}$ (otherwise
	\(ad\) would be and edge of \(\Theta_5\), but not \(ac\)).
	We first show that $d$ is below $br_m'$.
	Recall that $P_{ab}$ is
	fixed with respect to $T_{ba}$. Since $d$ is outside of $T_{ab}$ and $P_{ab}$,
	if $p_4p_0$ is inside $T_{ab}$, $d$ must be below $br_m'$. Since the
	slope of $p_0p_4$ is less than the slope of $\ell a$, it is sufficient to
	show that $p_4$ is inside $T_{ab}$ which follows by
	Claim~\ref{claim:pentagon}.
	By Observation \ref{obs} we have that $0\leq \angle
	ab\ell' \leq \pi/10$. Thus we can map the  points $(\qs,\qt,\qw,\qu,\qv)$ to
	$(b,a,d,r_m',\ell')$ and apply
	Lemma~\ref{t334}. Thus $|bd|+K|da|\leq K|ab|$ for $K\geq 5.70$.
	Our induction hypothesis and Lemma~\ref{t334} imply that there
	is a path from $b$ to $a$ with length at most
	\begin{displaymath}
		|\pathF{a,b}| \leq |bd|+|\pathF{d,a}|\leq |bd| + K|da|\leq K|ab|.
	\end{displaymath}

\end{proof}

\begin{lemma}\label{case8}
	If $ac$ and $bd$ cross and both $c$ and $d$ are in $P_{ab}$,
	then $|\pathF{a,b}| \leq K|ab|$ for $K \geq 6.16$.
\end{lemma}

\begin{proof}
	(Figures~\ref{examplet5s:sketch},~\ref{acbdcross})
	We show $\Phi = |ac| + K|cd| + |db| - K|ab| \leq 0$, which implies
	$|\pathF{a,b}| \leq  |ac| + |\pathF{c,d}| + |db| \leq K |ab|$ by the triangle
	inequality and the induction hypothesis.
	Under Transformation~\ref{transformt5}, Lemma~\ref{posofc} implies that
	$\Phi$ is maximized when $\alpha = \pi/10$, so we assume this is the case.
	Since $c$, $d$, and $P_{ab}$ are fixed, $c$ and $d$ are still inside $P_{ab}$
	after Transformation~\ref{transformt5}.
	Given that $c$ and $d$ are in $P_{ab}$, the furthest apart $c$ and $d$ can be
	is if they are both on a diagonal of $P_{ab}$. The length of one
	side of $P_{ab}$ is at most
	$\frac{\sin(\pi/10)}{\sin(3\pi/10)}|ab|$. That means a diagonal of
	$P_{ab}$, and thus $|cd|$, has length at most
	$2\sin(3\pi/10)\frac{\sin(\pi/10)}{\sin(3\pi/10)}|ab| = 2\sin(\pi/10)|ab|$.
	At their longest, $|ac|$ and $|bd|$ each have length
	$\frac{\sin(2\pi/5)}{\sin(3\pi/10)}|ab|$ by the law of sines.
	We want
	\begin{displaymath}
		\Phi = |ac|+K|cd|+|db| -K|ab| \leq 0.
	\end{displaymath}
	Solving for $K$ gives
	\begin{displaymath}
	K
	\geq \frac{|ac|+|db|}{|ab|-|cd|}
	\geq \frac{2\cdot\sin(2\pi/5)}{\sin(3\pi/10)\cdot(1-2\cdot\sin(\pi/10))}
	= 6.15\ldots 
	\end{displaymath}
\end{proof}

\section{Proving a spanning ratio of \texorpdfstring{$5.70$}{5.70}}\label{down2K}

In this section we present a lemma with a stronger bound for the case handled
by Lemma~\ref{case8}.
Proving this lemma requires a careful analysis of the
locations of $c$ and $d$ and the tradeoffs between the values of $|ac|+|db|$
and $K|cd|$. Let $\Phi = |ac|+K|cd|+|db|-K|ab|$. For the rest of this section,
assume we have applied Transformation~\ref{transformt5}, and thus $\alpha =
\pi/10$ and $\Phi$ is maximized. Since $P_{ab}$, $c$ and $d$ are fixed, both
$c$ and $d$ are still in $P_{ab}$. Let $c'$ be the intersection of the line
through $a$ and $c$ and the segment $p_0p_1$, and let $d'$ be the intersection
of the line through $b$ and $d$ and the segment $p_3p_4$. See
Figure~\ref{c'd'triangle}. Let $\Phi' = |ac'|+K|c'd'|+|d'b|-K|ab|$, and let
$\Phi'' = |ap_1|+K|p_1p_3|+|p_3b|-K|ab|$.
\begin{wrapfigure}[9]{l}{0.4\textwidth}
	\centering\vspace{-0.2 cm}
	\includegraphics[page=18,scale=.7]{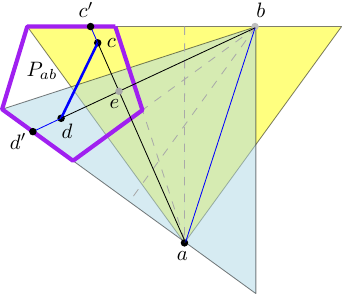}
	\caption{Points $c'$ and $d'$.}%
	\label{c'd'triangle}
\end{wrapfigure}
We split the analysis into three steps that amount to proving the following
lemmas:
\vspace{-1ex}
\begin{restatable}{lemma}{LemmaMainOne}\label{mainlemma1}
	For all $K\geq5.70$, $\Phi \leq \Phi'$.
\end{restatable}
\vspace{-1.9ex}
\begin{restatable}{lemma}{LemmaMainTwo}\label{mainlemma2}
	For all $K\geq5.70$, $\Phi' \leq \Phi''$.
\end{restatable}
\vspace{-1.9ex}
\begin{restatable}{lemma}{LemmaIntermediate}\label{intermediate}
	For all $K\geq5.70$, $\Phi''\leq 0$.
\end{restatable}

The following lemma follows from these lemmas, the triangle inequality, and
the induction hypothesis. It supersedes Lemma~\ref{case8}:
\begin{lemma}\label{newandimproved}
	If $ac$ and $bd$ intersect and both $c$ and $d$ are in $P_{ab}$, then $|\pathF{a,b}|\leq K|ab|$ for $K\geq5.70$.
\end{lemma}
Substituting Lemma~\ref{newandimproved} for Lemma~\ref{case8} in
the proof of Theorem \ref{main} brings the spanning ratio of the
$\Theta_5$-graph down to $5.70$. We are left with proving
Lemmas~\ref{mainlemma1},~\ref{mainlemma2}, and~\ref{intermediate}, which is done in the next three subsections.
\begin{figure}
	\begin{subfigure}[b]{0.5\textwidth}
		\centering
		\includegraphics[page =25, scale = 0.8]{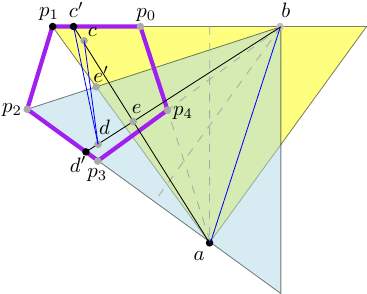}
		\caption{Proof that $|c'd|\geq |cd|$.}%
		\label{cdprime1}
	\end{subfigure}%
	\begin{subfigure}[b]{0.5\textwidth}
		\centering
		\includegraphics[page =31, scale =0.8]{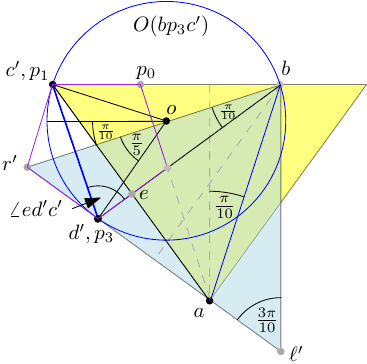}
		\caption{Maximum of $\angle ed'c'$.}%
		\label{cdprime2}
	\end{subfigure}
	\caption{Finding the longest distance from $a$ to $b$ when $c$ and $d$ are in $P_{ab}$,}
\end{figure}

\subsection{Proof of Lemma \ref{mainlemma1}}
Lemma~\ref{mainlemma1} states that $|ac|+|bd|+K|cd|-K|ab|\leq |ac'|+|bd'|+K|c'd'|-K|ab|$ for $K\geq 5.70$. See Figure~\ref{cdprime1}. Let $e$ be the intersection of $ac$ and $bd$, and let $e'$ be the intersection of $br'$ and $a\ell$. Observe that $\angle \ell e' r' =2\pi/5$, and thus we can see that $\angle dec\geq 2\pi/5$. This implies that $\angle dec$ cannot be the smallest angle in $\triangle dec$, since that would require $\angle dec\leq \pi/3$. Thus at least one of $\angle dce$ and $\angle edc$ is the smallest angle in $\triangle dec$. Since we have applied Transformation \ref{transformt5}, and can thus assume that $\alpha = \pi/10$, the cases are symmetric. We can therefore, without loss of generality, assume that $\angle dce$ is the smallest angle in $\triangle dec$.

\restate{\LemmaMainOne*}

\begin{proof}
	Since $c$ lies on $ac'$ and $d$ lies on $bd'$, we have $|ac|\leq |ac'|$ and $|bd|\leq |bd'|$, and it is sufficient to show that $|cd|\leq |c'd'|$. We first show that $|cd|\leq |c'd|$. Since $\angle dce$ is the smallest angle in $\triangle dec$, $\angle dce< \pi/3$. That implies that $\angle c'cd >\pi/2$, which implies that $c'd$ is the longest side of triangle $\triangle cc'd$, and thus $|cd|\leq |c'd|$. See Fig.\ref{cdprime1}.
	
	We now show that $|c'd'|\geq |c'd|$. If $\angle c'dd' \geq \pi/2$, then $c'd'$ is the longest side of $\triangle c'dd'$, and $|c'd'|\geq |c'd|$ and we are done. Otherwise assume $\angle c'dd'< \pi/2$.
	
	The law of sines tells us that $\frac{|c'd'|}{\sin \angle c'dd'} =
	\frac{|c'd|}{\sin \angle dd'c'}$. Since $\sin \theta$ is an increasing
	function for $0\leq \theta <\pi/2$, showing that $\angle c'dd' \geq \angle
	dd'c'$ is sufficient to show $|c'd'|\geq |c'd|$, as it would imply both
	angles are $<\pi/2$. Observe that $\angle c'dd'\geq \angle c'ed'$ and
	$\angle ed'c'=\angle dd'c'$, thus it is sufficient to prove that $\angle
	c'ed'\geq \angle ed'c'$.
	
	Observe that $\angle ced =\angle c'ed' \geq 2\pi/5$. We now find the
	maximum of $\angle dd'c' = \angle ed'c'\leq 2\pi/5$. Observe that if $c'$
	moves left, $\angle ed'c'$ increases, thus assume $c'$ is at $p_1$. Let
	$O(bp_3c')$ be the circle through $b$, $p_3$, and $c'$ with center $o$.
	Observe that $o$ lies on $br'$. Observe that $\angle r'bd' = \pi/10$, thus
	$\angle r'op_3 = \pi/5$. Segment $or'$ makes an angle of $\pi/10$ with the
	horizontal line through $o$. Thus $od'$ makes an angle of $3\pi/10$ with
	the horizontal line through $o$, and thus the line tangent to $O(bp_3c')$
	at $p_3$ is the line supporting $\ell'r'$, since $\ell'r'$ makes an angle
	of $3\pi/10$ with the vertical line through $\ell'$. See
	Figure~\ref{cdprime2}. That implies that $[p_2,p_3)$ lies outside of
	$O(bp_3c')$, which means for every point $d'$, $\angle ed'c' \leq \angle
	ep_3c' = 2\pi/5$, and thus $\angle c'dd' \geq \angle dd'c'$ as required.
\end{proof}

\begin{figure}
	\centering
	\includegraphics[page=22,scale = 0.85]{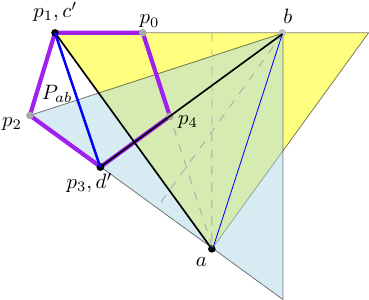}
	\caption{An example of $\Phi''$.}%
	\label{worstcase}
\end{figure}

\subsection{Proof of Lemma \ref{mainlemma2}}
Observe that $|ap_1|+K|p_1p_3|+|p_3b| = |ap_0|+K|p_0p_2|+|p_2b|$ when $\alpha =
\pi/10$, since $T_{ab}$ and $T_{ba}$ are the same size and the cases are
symmetric.
We prove that
\begin{align*}
	\Phi' = |ac'|+K|c'd'|+|d'b|-K|ab| &\leq |ap_1|+K|p_1p_3|+|p_3b| -K|ab| = \Phi'' \\
	&= |ap_0|+K|p_0p_2|+|p_2b|-K|ab|.
\end{align*}

\begin{figure}
	\begin{subfigure}[b]{0.45\textwidth}
		\centering
		\includegraphics[page = 23]{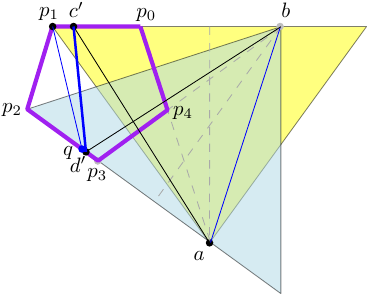}
		\caption{The point $q$ such that $|p_1q|=|c'd'|$ lies between $d'$ and $p_2$.}%
		\label{worstcase2}
	\end{subfigure}%
	\hfill
	\begin{subfigure}[b]{0.45\textwidth}
		\centering
		\includegraphics[page = 24]{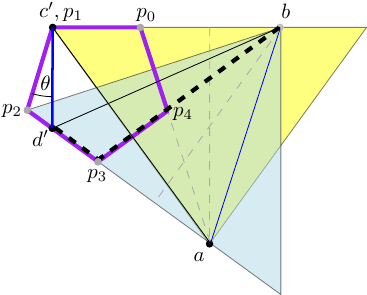}
		\caption{We look at the change in $|d'p_3|+K|c'd'|$ with respect to $\theta$.}%
		\label{final}
	\end{subfigure}
	\caption{}
\end{figure}

\restate{\LemmaMainTwo*}

\begin{proof}
	
	(Figure~\ref{worstcase})
	Without loss of generality, we assume that $|p_1c'|\leq|p_2d'|$. We
	show that $\Phi'$ is maximized when $c'=p_1$ and $d'=p_3$.
	
	(Figure~\ref{worstcase2})
	Observe that $|p_1p_2|\leq |c'd'|\leq |p_1p_3|$.
	Let $z$ be a point on $p_2p_3$ that moves from $p_2$ to $p_3$, and observe
	that $|p_1z|$ takes on every value from $|p_1p_2|$ to $|p_2p_3|$. Thus
	there must be a point $q$ on $p_2p_3$ such that $|p_1q|=|c'd'|$.
	
	We claim that $|ap_1|+|bq| \geq |ac'|+|bd'|$, which
	implies that $\Phi' \leq |ap_1| + K|p_1q| + |qb| - K|ab|$.
	
	Observe $|ap_1|\geq |ac'|$, since $\angle p_1c'a >\pi/2$, making $ap_1$ the
	longest edge in triangle $\triangle ac'p_1$. We claim that $q$ is between
	$d'$ and $p_2$, and thus $|bq|\geq |bd'|$ since $\angle bd'q > \pi/2$.
	By contradiction, assume that $q$ is between $d'$ and
	$p_3$. Since $|p_1c'|\leq|p_2d'|$, $\angle qd'c' >\pi/2$, which implies
	that $|c'q|>|c'd'|$. Also note that $\angle qd'p_1 >\pi/2$, which implies
	$|p_1q|>|c'q|>|c'd'|$, a contradiction. Thus assuming that $c'=p_1$ and
	$d'=q$ does not decrease $\Phi'$.
	
	Now, given that $c'$ is on $p_1$, we show that $\Phi' \leq
	|ac'|+|bp_3|+K|c'p_3|$, that is, when $d'$ is on $p_3$. To do this we
	define another function $\Phi^* = |ac'|+|d'p_3|+|p_3b|+K|c'd'|-K|ab|$. See
	Figure~\ref{final}. Since $|bd'|\leq |d'p_3|+|p_3b|$ by the triangle
	inequality, $\Phi' \leq \Phi^*$, and observe that $\Phi'
	=\Phi^*=\Phi''$ when $d'=p_3$. We show that $\Phi^*$ is maximized
	when $d'=p_3$, thus implying that $\Phi'$ is also maximized when $d'=p_3$,
	and $\Phi'\leq \Phi''$. Let $\theta = \angle p_2p_1d'$. We allow $d'$ to
	move along $p_2p_3$ until $d'$ is on $p_3$, and fix all other points, and
	observe how $\Phi^*$ changes with $\theta$.
	
	We first rewrite $\Phi^*$ as $\Phi^*=
	|ac'|+|p_2p_3|-|p_2d'|+|p_3b|+K|c'd'|-K|ab|$. Using the sine law we get
	$|p_2d'|=\frac{\sin\theta}{\sin(2\pi/5-\theta)}|p_1p_2|$, and $|c'd'| =
	\frac{\sin(3\pi/5)}{\sin(2\pi/5-\theta)}|p_1p_2|$. All other terms of
	$\Phi^*$ have fixed values with respect to $\theta$. Thus
	
	\begin{align}
		\frac{d\Phi^*}{d\theta}
		&= \frac{d}{d\theta}\left( K\frac{\sin(3\pi/5)}{\sin(2\pi/5-\theta)}|p_1p_2| -\frac{\sin\theta}{\sin(2\pi/5-\theta)}|p_1p_2|\right)\nonumber\\
		&= \frac{K\cos(2\pi/5-\theta)\sin(3\pi/5)-\cos\theta\sin(2\pi/5-\theta) - \sin\theta\cos(2\pi/5-\theta)}{\sin^2(2\pi/5-\theta)}|p_1p_2|\nonumber\\
		&= \frac{K\cos(2\pi/5-\theta)\sin(3\pi/5)-\sin(2\pi/5)}{\sin^2(2\pi/5-\theta)}|p_1p_2|.\label{number}
	\end{align}
	
	Observe that $0\leq \theta \leq 3\pi/10$. The denominator of \eqref{number}
	is always positive. The numerator of \eqref{number} is minimized at $\theta
	= 0$, which for $K\geq5.70$ is positive. Thus \eqref{number} is always
	positive for $0\leq \theta\leq 3\pi/10$, thus $\Phi^*$ is increasing in
	$\theta$, and is maximized when $d'=p_3$, as required.
	Thus $\Phi' \leq \Phi^* \leq \Phi'' = |ap_1|+K|p_1p_3|+|p_3b|-K|ab|$ as required.
\end{proof}

\newpage
\subsection{Proof of Lemma \ref{intermediate}}

\restate{\LemmaIntermediate*}

\begin{proof}
	(Figure~\ref{worstcase})
	We apply Transformation~\ref{transformt5} with
	\(\alpha = \frac{\pi}{10}\) and assume that $|ab|=1$. Then using the law of
	sines we get $|bp_3| = 1$, $|ap_1|= \frac{\sin(2\pi/5)}{\sin(3\pi/10)}$,
	and $|p_1p_3|=2\sin(3\pi/10)\frac{\sin(\pi/10)}{\sin(3\pi/10)} =
	2\sin(\pi/10)$. We want
	\begin{displaymath}
		\Phi'' = |ap_1| + K|p_1p_3| + |p_3b| - K|ab| \leq 0.
	\end{displaymath}
	Solving for \(K\) gives
	\begin{displaymath}
		K \geq
		\frac{|ap_1|+|p_3b|}{|ab|-|p_1p_3|}
		= \frac{\frac{\sin(2\pi/5)}{\sin(3\pi/10)} +1}{1-2\sin(\pi/10)}
		= 5.69\ldots 
	\end{displaymath}
\end{proof}

\section{Open Problems}\label{conclusion}

\ifisaac\else%
Using a few simple geometric observations and arguments, we have lowered the
spanning ratio of $\Theta_5$ from $9.96$ to $5.70$, bringing us closer to the
lower bound of $3.798$ and thus a tight bound.
\fi
The obvious open problem that remains is closing the gap between the upper and lower bound on the spanning ratio of the $\Theta_5$-graph.

\paragraph{\large{Acknowledgements:}} We thank Elena Arseneva for many fruitful discussions on the topic. 

\bibliographystyle{plain}
\bibliography{bibliography}

\begin{thebibliography}{10}

\bibitem{theta3connected}
O.~Aichholzer, S.~W. Bae, L.~Barba, P.~Bose, M.~Korman, A.~van Renssen,
  P.~Taslakian, and S.~Verdonschot.
\newblock {T}heta-3 is connected.
\newblock {\em Computational geometry}, 47(9):910--917, 2014.

\bibitem{barbat4}
L.~Barba, P.~Bose, J.~L. De~Carufel, A.~van Renssen, and S.~Verdonschot.
\newblock On the stretch factor of the {T}heta-4 graph.
\newblock In {\em Proceedings of the 13th International Symposium on Algorithms
  and Data Structures (WADS)}, pages 109--120, 2013.

\bibitem{soda19}
P.~Bose, J.~L.~De Carufel, D.~Hill, and M.~Smid.
\newblock On the spanning and routing ratio of {T}heta-four.
\newblock In {\em Proceedings of the Thirtieth Annual {ACM-SIAM} Symposium on
  Discrete Algorithms (SODA)}, pages 2361--2370, 2019.

\bibitem{morenotbetter}
P.~Bose, J.~L.~De Carufel, P.~Morin, A.~van Renssen, and S.~Verdonschot.
\newblock Towards tight bounds on {T}heta-graphs: More is not always better.
\newblock {\em Theoretical Computer Science}, 616:70--93, 2016.

\bibitem{sander5}
P.~Bose, P.~Morin, A.~van Renssen, and S.~Verdonschot.
\newblock The {T}heta-5-graph is a spanner.
\newblock {\em Computational Geometry}, 48(2):108--119, 2015.

\bibitem{chewtd}
L.~P. Chew.
\newblock There are planar graphs almost as good as the complete graph.
\newblock {\em Journal of Computer and System Sciences}, 39(2):205 -- 219,
  1989.

\bibitem{clarkson1987}
K.~Clarkson.
\newblock Approximation algorithms for shortest path motion planning.
\newblock In {\em Proceedings of the Nineteenth Annual ACM Symposium on Theory
  of Computing (STOC)}, pages 56--65, 1987.

\bibitem{keiltheta}
J.~M. Keil.
\newblock Approximating the complete {E}uclidean graph.
\newblock In {\em Proceedings of the 1st Scandinavian Workshop on Algorithm
  Theory (SWAT)}, pages 208--213, 1988.

\bibitem{keil1992}
J.~M. Keil and C.~A. Gutwin.
\newblock Classes of graphs which approximate the complete {E}uclidean graph.
\newblock {\em Discrete {\&} Computational Geometry}, 7(1):13--28, 1992.

\bibitem{nawarphd}
N.~M.~El Molla.
\newblock {\em Yao spanners for wireless ad-hoc networks}.
\newblock PhD thesis, Villanova University, 2009.

\bibitem{ruppert}
J.~Ruppert and R.~Seidel.
\newblock Approximating the d-dimensional complete {E}uclidean graph.
\newblock In {\em Proceedings of the 3rd Canadian Conference on Computational
  Geometry (CCCG)}, 1991.

\end{thebibliography}

\end{document}